\documentclass[12pt]{huber_article}

\usepackage{amsmath,amsthm,amsfonts}
\usepackage{dsfont}
\usepackage{booktabs}

\newcommand{\prob}{\mathbb{P}}
\newcommand{\real}{\mathbb{R}}
\newcommand{\mean}{\mathbb{E}}
\newcommand{\poisdist}{\textsf{Pois}}
\newcommand{\invgammadist}{\textsf{InvGamma}}
\newcommand{\expdist}{\textsf{Exp}}
\newcommand{\gammadist}{\textsf{Gamma}}
\newcommand{\unifdist}{\textsf{Unif}}
\newcommand{\normaldist}{\textsf{N}}
\newcommand{\betadist}{\textsf{Beta}}
\newcommand{\berndist}{\textsf{Bern}}
\newcommand{\iid}{\mathrel{\mathop{\sim}\limits^{\textrm{iid}}}}
\newcommand{\tpa}{{\sffamily{TPA}}}
\newcommand{\gpas}{{\sffamily{GPAS}}}
\newcommand{\ind}{{\mathds{1}}}

\newcommand{\gk}{{\arabic{line})}\stepcounter{line}}
\newcounter{line}

\newtheorem{lemma}{Lemma}
\newtheorem{fact}{Fact}
\theoremstyle{definition}
\newtheorem{definition}{Definition}
\newtheorem{example}{Example}

\begin{document}

\title{\noindent An estimator for Poisson means whose relative
error distribution is known}

\author{Mark Huber}

\maketitle

\begin{abstract}
Suppose that $X_1,X_2,\ldots$ are a stream of independent, identically
distributed Poisson random variables with mean $\mu$.  This work presents
a new estimate $\mu_k$ for $\mu$ with the property that the distribution
of the relative error in the estimate ($(\hat \mu_k/\mu) - 1$) is known,
and does not depend on $\mu$ in any way.  This enables the construction
of simple exact confidence intervals for the estimate, as well as a means
of obtaining fast approximation algorithms for high dimensional integration
using \tpa{}.  The new estimate requires a random number of Poisson draws,
and so is best suited to Monte Carlo applications.
As an example of such an application, the method is applied to obtain an 
exact confidence interval for the normalizing constant of the Ising model.
\end{abstract}



  \section{Introduction}

  A random variable $X$ 
  is Poisson distributed with mean $\mu$ (write $X \sim \poisdist(\mu)$) if 
  $\prob(X = i) = \exp(-\mu)\mu^i/i!$ for $i \in \{0,1,2,\ldots\}$.
  Suppose that $X_1,X_2,\ldots$ are independent identically distributed (iid)
  Poisson random variables 
  with mean $\mu$.  The purpose of this paper is to present a new
  estimator for $\mu$ that uses almost the ideal number of 
  Poisson draws.

  Our estimate will not only use draws from 
  $X_1,X_2,\ldots \iid \poisdist(\mu)$, but make extra random choices as
  well.  This external source of randomness can be represented by a
  random variable $U$ that is uniformly distributed over $[0,1]$ (write
  $U \sim \unifdist([0,1])$.  As is well known, a single draw $U$ is 
  equivalent to an infinite number of draws 
  $U_1,U_2,\ldots \iid \unifdist([0,1])$.

  \begin{definition}  Suppose $\cal A$ is a computable function
  of $X_1,X_2,\ldots \iid \poisdist(\mu)$ and auxiliary randomness (represented
  by $U \sim \unifdist([0,1])$ that outputs $\hat \mu$.  Let $T$ be 
  a stopping time with respect to the natural filtration so that the value of
  $\hat \mu$ only depends on $U$ and $X_1,\ldots,X_T$.  Then call $T$ the 
  {\em running time} of the algorithm.
  \end{definition}

  \begin{definition}
  For an estimate $\hat \mu$ of $\mu$, the {\em relative error} is 
  $\epsilon_\text{relative} = (\hat \mu/\mu) - 1$.
  Call $\hat \mu$ an {\em $(\epsilon,\delta)$-approximation} for 
  $\mu$ if $\prob(|\epsilon_{\text{relative}}| > \epsilon) < \delta$.
  \end{definition}

  The simplest algorithm for estimating $\mu$ just fixes $T = n$,
  and sets 
  \[
  \hat \mu_n = \frac{X_1 + \cdots + X_n}{n}.
  \]
  This basic estimate has several good properties.  First, it is unbiased,
  that is, $\mean[\hat \mu_n] = \mu$.  Second, it is consistent, as 
  $n \rightarrow \infty$, $\hat \mu_n \rightarrow \mu$ with probability 1.
  Third, it is efficient.  Using the Fisher information about $\mu$
  contained in a single $X_i$ with the Cr\'amer-Rao inequality, it is
  possible to show that this estimate has the minimum variance of any
  unbiased estimate that only uses $n$ draws.

  However, this estimate is difficult to use for building
  $(\epsilon,\delta)$-approximation algorithms, as the ratio $\hat \mu_n/\mu$
  depends strongly on $\mu$.  It is well known that 
  $X_1 + \cdots + X_n \sim \poisdist(n\mu)$. 
  Using techniques such as Chernoff bounds 
  to bound the tail of a Poisson distribution, it is possible to bound
  the value of $n$ needed to get an $(\epsilon,\delta)$-approximation.
  
  These bounds however are not tight, and inevitably a slightly larger
  value of $n$ than is necessary will be needed to meet the $(\epsilon,\delta)$
  requirements.

  The goal of this work is to introduce a new estimate for the mean of 
  the Poisson distribution whose relative error is independent of $\mu$,
  the quantity being estimated.

  \subsection{Examples of estimates whose relative error is independent
  of the parameter}

  As an example of a distribution where the basic estimate is scalable, 
  say that $Z$ is normally distributed with mean $\mu$ and variance $\sigma^2$
  (write $Z \sim \normaldist(\mu,\sigma^2)$) if $Z$ has density
  $f_Z(s) = (2\pi \sigma^2)^{-1/2} \exp(-(s - \mu)^2/[2\sigma^2]).$  As is
  well known, normals can be scaled and shifted, and still remain normal.

  \begin{fact}
  For $Z \sim \normaldist(\mu,\sigma^2)$ and constants $a$ and $b$,
  $aZ + b \sim \normaldist(a\mu + b,a^2 \sigma^2)$.
  \end{fact}

  Now consider $Z_1,Z_2,\ldots \iid \normaldist(\mu,\mu^2)$.  
  In this case, the sample average
  satisfies $\hat \mu_n \sim \iid \normaldist(\mu,\mu^2/n)$, and 
  $(\hat \mu_n/\mu) - 1 \sim \normaldist(0,1/n)$.  
  Note that the distribution of 
  the relative error does not depend in any way on the parameter $\mu$
  being estimated.

  For another example, say that $Y$ is exponentially distributed with
  rate $\mu$ (write $Y \sim \expdist(\mu)$) if 
  $Y$ has density $f_Y(s) = \mu \exp(-\mu s) \ind(s \geq 0)$.  Here
  $\ind(\cdot)$ is the indicator function that is 1 when the argument inside
  is true, and 0 when it is false.
  As with normals, scaled exponentials are still exponential.  Unlike
  normals, the rate parameter is divided by the scale.

  \begin{fact}
  For $Y \sim \expdist(\mu)$ and constant $a$, 
   $a Y \sim \expdist(\mu/a)$.
  \end{fact}

  Say that $T$ has a Gamma distribution with shape parameter $k$ and 
  rate parameter $\mu$ (write $T \sim \gammadist(k,\mu)$) if 
  $T$ has density 
  \[
  f_T(s) = \mu^k \Gamma(k)^{-1} s^{k-1}e^{-\mu s} \ind(s \geq 0).
  \]

  Adding iid exponentially distributed random variables 
  together gives a Gamma distributed random variable.

  \begin{fact}
  \label{FCT:sumexponentials}
  If $Y_1,Y_2,\ldots \iid \expdist(\mu)$, then for any $k$,
  $Y_1 + \cdots + Y_k \sim \gammadist(k,\mu)$.
  \end{fact}

%

  Given $n$ draws $Y_1,\ldots,Y_n$, 
  the maximum likelihood estimator for $\mu$ in this context is the 
  inverse of the sample average (see for instance~\cite{ramachandrant2009}):
  \[
  \hat \mu_{\text{MLE},n} = \frac{n}{Y_1 + \cdots + Y_n}. 
  \] 
  That gives
  \[
  \frac{\hat \mu_{\text{MLE},n}}{\mu} = \frac{n}{\mu Y_1 + \cdots + \mu Y_n}.
  \]

  By scaling $\mu Y_1 \sim \expdist(\mu/\mu) = \expdist(1)$, so 
  $\mu Y_1 + \cdots + \mu Y_n \sim \gammadist(n,1)$.
  Therefore the relative error in $\hat \mu_{\text{MLE},n}$ is independent
  of $\mu$!

  Now, the distribution of $1/T$ where $T \sim \gammadist(k,\mu)$ is 
  called an Inverse Gamma distribution with shape parameter $k$ and 
  scale parameter $\mu$ (write $1/T \sim \invgammadist(k,\mu)$.  Note
  that what was the rate parameter $\mu$ for the Gamma becomes a scale
  parameter for the Inverse Gamma.  The mean of this $\invgammadist(k,\mu)$
  random variable is $\mu/(k-1)$.

  That means a unbiased estimate for $\mu$ is 
  \[
  \hat \mu_{\text{unbiased},n} = \frac{n - 1}{Y_1 + \cdots + Y_n}
  \]
  since the right hand side is $\invgammadist(n,(n-1)\mu)$.

  What about discrete variables that are inherently unscalable?  
  In~\cite{hubertoappearc}, the author presented a method for turning
  a stream of iid Bernoulli random variables (which are 1 with probability $p$,
  and 0 with probability $1 - p$) into a $\gammadist(k,p)$ random variable,
  where $k$ is a parameter chosen by the user.
  This could then be used with the known relative error 
  estimate for exponentials to obtain a known relative error estimate
  for Bernoullis.

  While the Bernoulli application has the widest use, Poissons do appear
  in the output of a Monte Carlo approach to high dimensional integration
  called the Tootsie Pop Algorithm (\tpa{})~\cite{huber2010b,huber2014a}.
  Therefore, to use \tpa{} to build $(\epsilon,\delta)$-approximation
  algorithms, it is 
  useful to have a known relative error distribution for 
  Poisson random variables.

  The remained of this paper is organized as follows.  
  Section~\ref{SEC:method} describes the new estimate and why it works.
  It also bounds the expected running time.  Section~\ref{SEC:applications}
  then shows how this procedure can be used together with \tpa{} to 
  obtain $(\epsilon,\delta)$-approximations for normalizing constants of
  distributions.

  \section{The method}
  \label{SEC:method}

  The new estimate is based upon properties of Poisson point processes.

  \begin{definition}
    A Poisson point process of rate $\mu$ 
    on $\real$ is a random subset $P \subset \real$
    such that the following holds.
    \begin{itemize}
      \item
      For all $a \leq b$, $\mean[\#(P\cap [a,b])] = \mu(b - a)$.
      \item
      For all $a \leq b \leq c \leq d$, $\#(P \cap [a,b])$ and 
    $\#(P \cap [c,d])$ are independent.
    \end{itemize}
  \end{definition}

  It is well known that there
  are (at least) 
  two ways to construct a Poisson point process, which forms the basis
  of the estimate.  

  The first method for simulating a Poisson point process is to
  take advantage of the fact that the number of points within a given
  interval has a Poisson distribution.

  \begin{fact}
  \label{FCT:uniform}
    Let $P$ be a Poisson point process of rate $\mu$.  Then
    for all $a \leq b$, $\#(P \cap [a,b]) \sim \poisdist(\mu(b - a))$.
    Moreover, conditioned on the number of points in the interval,
    the points themselves are uniformly distributed over the interval.  
    That is,
    \[
    [P \cap [a,b]|\#(P \cap [a,b]) = n] \sim \unifdist([a,b]^n).
    \]
  \end{fact}

  The second method to building a Poisson point process of rate $\mu$ is 
  to use the fact that the distances between successive points are 
  iid exponentially distributed with rate $\mu$.

  \begin{fact}
  Let $P$ be a Poisson point process of rate $\mu$.  Also, let
  $P \cap [0,\infty) = \{P_1,P_2,\ldots\}$ where $P_i \leq P_{i+1}$ for all
  $i$.  Setting $A_i = P_{i+1} - P_i$ (and $A_1 = P_1$), we have that
  $A_1,A_2,\ldots \iid \expdist(\mu)$.
  \end{fact}

  Using Fact~\ref{FCT:sumexponentials}, $P_k$ will have a Gamma distribution
  with shape parameter $k$ and rate parameter $\mu$.  So, this is how the 
  estimate works.  First, generate $N_1$, the number of points of the Poisson
  point process in $[0,1]$.  If this is at least $k$, then we know that 
  $P_k \in [0,1]$.  Otherwise, generate $N_2$, the number of points
  in $[1,2]$.  If $N_1 < k$ and $N_1 + N_2 \geq k$, then $P_k \in [0,2]$.
  Otherwise, keep going, generating more Poisson random variates until
  we know that $P_k \in [i,i+1]$ for some integer $i$.

  Let $A = N_1 + \cdots + N_{i-1}$.  Then we 
  know that $A < k$ points are in $[0,i]$, and $A + N_i \geq k$.  From 
  Fact~\ref{FCT:uniform}, the $N_i$ points are uniformly distributed over
  $[i,i+1]$.  The $k - A$ smallest of these points will be $P_k$.  One more
  well known fact about the order statistics of uniform random variables
  will be helpful.

  \begin{fact}
  If $U_1,\ldots,U_n \iid \unifdist([0,1])$, then 
  $U_{(i)} \sim \betadist(i,n-i+1).$
  \end{fact}

  Putting this all together gives the following estimate, called
  the Gamma Poisson Approximation Scheme, or \gpas{} for short.

  \begin{center}
\setcounter{line}{1}
\begin{tabular}{rll}
\toprule
\multicolumn{3}{l}{{\sc Gamma\_Poisson\_Approximation\_Scheme}} \\
\multicolumn{3}{l}{
 {\em Input:} 
 $k$ \quad
  {\em Output:} 
 $\hat \mu_k$
 } \\
\midrule
\gk & \hspace*{0em} $A \leftarrow 0, i \leftarrow 0$ \\
\gk & \hspace*{0em} While $A < k$ & [Draw $k$ points.]\\
\gk & \hspace*{1em}   $T \leftarrow \poisdist(\mu)$ \\
\gk & \hspace*{1em}   If $A + T \geq k$ & [Then have $k$ points.] \\
\gk & \multicolumn{2}{l}
  {\hspace*{2em}   $T' \leftarrow i + \betadist(k - A,T - (k - A) + 1)$} \\
\gk & \multicolumn{2}{l} 
  {\hspace*{1em}  $A \leftarrow A + T$, $i \leftarrow i + 1$} \\
\gk & \hspace*{0em} $\hat \mu_k \leftarrow (k - 1)/T'$ \\
\bottomrule
\end{tabular}
\end{center}

\begin{lemma}
  The expected number of Poisson random variables drawn by GPAS is 
  bounded above by $1 + k/\mu$.  
\end{lemma}

  \begin{proof}
    The number of Poisson random variables drawn is 
    $\lceil P_k \rceil \leq P_k + 1$.  
    Since $P_k \sim \gammadist(k,\mu)$, $\mean[P_k] = k/\mu$, which shows
    the result.
  \end{proof}

Note that any fixed time algorithm would need a similar number of samples
to obtain such a result. 

\begin{fact}
  The Fisher information of $\mu$ for $X \sim \poisdist(\mu)$ is $1/\mu$.
\end{fact}

Therefore, by the Cr\'amer-Rao inequality, the variance of any unbiased
estimate $\hat \mu$
that uses $n$ draws is at least $\mu/n$, so for 
$k/\mu$ draws, the standard deviation will be at least $\mu/\sqrt{k}$.

\begin{lemma}
  The output $\hat \mu_k$ of \gpas{} has distribution 
  $\invgammadist(k,(k-1)\mu)$,
  and has standard deviation $\mu/\sqrt{k-2}$.
\end{lemma}

Therefore to first order for the same number of samples, the resulting
unbiased estimate achieves the minimum variance.  Of course, the real
benefit of using \gpas{} is that is provides an exact relative error
distribution, thus allowing for precise calculations of the chance of 
error.

\begin{example}
  For $k = 1000$, \gpas{} is a $(0.1,0.0018)$-approximation algorithm
  for $\mu$.
  \begin{align*}
  \prob((1-0.1)\mu \leq 999/T' \leq (1+0.1)\mu) &= 
    \prob(999/0.9 \geq T'/\mu \geq 999/1.1) \\
    &= 0.001786\ldots.
  \end{align*}
  since $T' \mu \sim \gammadist(1000,999).$
\end{example}

\begin{example}
What should $k$ be in order to make \gpas{} an $(0.1,10^{-6})$-approximation
algorithm?

Increasing the value of $k$ in the previous example until we reach the 
first place where $\prob((k-1)/0.9 \geq T'/\mu \geq (k-1)/1.1)$ gives
$k = 2561$ as the first place where this occurs.
\end{example}

In fact, in the previous example 
\[
\prob(2560/0.9 \geq T'/\mu \geq 2560/1.1) = 0.0000009970\ldots,
\]
and so is slightly smaller than the error bound requested. It is possible
to create an algorithm with exactly $10^{-6}$ chance of failure by running
\gpas{} either with $k = 2561$ or $k = 2560$ with the appropriate 
probabilities.  This gives the following algorithm, where 
$p_k$ is the 
cumulative
distribution function of a Gamma distribution with shape $k$ and 
rate $k - 1$.

 \begin{center}
\setcounter{line}{1}
\begin{tabular}{rl}
\toprule
\multicolumn{2}{l}{{\sc Exact\_GPAS}} \\
\multicolumn{2}{l}{
 {\em Input:} 
 $\epsilon,\delta$ \quad
  {\em Output:} 
 $\hat \mu$
 } \\
\midrule
\gk & \hspace*{0em} Let $f_i(s) = q_i(1/(1+\epsilon))+(1 - q_i(1/(1-\epsilon)))$
  \\
\gk & \hspace*{0em} Let $k \leftarrow 
  \min\{i:f_i(s) \leq \delta \}$ \\
\gk & \hspace*{0em} $p \leftarrow (\delta - f_{k}(s))/(f_{k-1}(s)-f_{k}(s))$ \\
\gk & Draw $C \leftarrow \berndist(p)$ \\
\gk & If $C = 1$ then $k \leftarrow k - 1$ \\
\gk & $\hat \mu \leftarrow $ \textsc{Gamma\_Poisson\_Approximation\_Scheme}$(k)$
  \\
\bottomrule
\end{tabular}
\end{center}

  \section{Applications}
  \label{SEC:applications}

So why approximate the mean of a Poisson in the first place?  One of the
applications is to the Tootsie Pop Algorithm 
(\tpa{})~\cite{huber2010b,huber2014a}.  Given a set
$A \subset B \in \real^n$, 
the purpose of \tpa{} is to estimate $\nu(B)/\nu(A)$ for
some measure $\nu$.

This is exactly the problem of approximating a high dimensional integral
that arises in such problems as finding the normalizing constant of 
a posterior distribution in Bayesian applications.
The output of \tpa{} (see~\cite{huber2010b,huber2014a}) is exactly a 
Poisson random variable with mean $\ln(\nu(B)/\nu(A))$. 

Typically the situation is that $\nu(A)$ is known, and the 
goal is to approximate the other.  Let $r = \ln(\nu(B)/\nu(A))$.
Then if $\hat r$ is an approximation for $r$, then
$\exp(\hat r)$ is an approximation for $\nu(B)/\nu(A)$, and 
$\nu(A) \exp(\hat r)$ is an approximation for $\nu(B)$.  

An $(\epsilon,\delta)$-approximation for $\nu(B)$ can therefore be obtained
by finding 
an $(\epsilon,\delta)$-approximation for $\exp(r)$.  Note
\begin{align*}
\prob((1-\epsilon)e^r \leq \exp(\hat r) \leq (1 + \epsilon)e^r)
 &= \prob(r + \ln(1 - \epsilon) \leq \hat r \leq r + \ln(1 + \epsilon)) \\
 &= \prob\left(1 + \frac{\ln(1 - \epsilon)}{r} \leq \frac{\hat r}{r}
   \leq 1 + \frac{\ln(1+\epsilon)}{r}\right).
\end{align*}

Since $|\ln(1 + \epsilon)| < |\ln(1 - \epsilon)|$, the
needed bound on 
the relative error is $\ln(1 + \epsilon)/r$.

A two-phase procedure is used to obtain the estimate.  In the first
phase, $r$ is estimated with a $(\epsilon,\delta/2)$-approximation called
$\hat r_1$.  So with probability at least $1 - \delta/2$, it holds that 
$r \geq \hat r_1/(1-\epsilon)$.  In
the second phase, $r$ is estimated with a 
$(\ln(1+\epsilon)\hat r_1^{-1}(1-\epsilon),\delta/2)$-approximation
called $\hat r_2$.

Using the union bound, the chance that both phases are successful is at
least $1 - \delta/2 - \delta/2 = 1 - \delta$, and the above calculation
shows that $\exp(\hat r_2)$ is an $(\epsilon,\delta)$-approximation
for $\exp(r)$.  The resulting algorithm can be given as follows.

 \begin{center}
\setcounter{line}{1}
\begin{tabular}{rl}
\toprule
\multicolumn{2}{l}{{\sc TPA\_Approximation\_Scheme}} \\
\multicolumn{2}{l}{
 {\em Input:} 
 $\epsilon,\delta$ \quad
  {\em Output:} 
 $\hat(\nu(B)/\nu(A))$
 } \\
\midrule
\gk & \hspace*{0em} $\hat r_1 \leftarrow$
 \textsc{Exact\_GPAS}
 $(\epsilon,\delta/2)$) \\
\gk & \hspace*{0em} $\hat r_2 \leftarrow$
 \textsc{Exact\_GPAS}
 $(\ln(1+\epsilon)\hat r_1^{-1}(1-\epsilon),\delta/2)$) \\
\gk & Output $\exp(\hat r)$ 
  \\
\bottomrule
\end{tabular}
\end{center}

This algorithm applies with the understanding that line 3 of 
the algorithm
\textsc{Gamma\_Poisson\_Approximation\_Scheme} is replaced with
$T \leftarrow$\tpa{}, that is, the Poisson with mean $\mu$ is replaced
by a call to \tpa{}.

\begin{center}
  \begin{table}[ht]
  \begin{center}
   \begin{tabular}{cccc}
   $\epsilon$ & $\delta$ & $\mean[T]$ for new method & $\mean[T]$ 
     from older method in~\cite{huber2014a} \\
   \midrule
     0.2 & 0.2 & $607\pm 5$   & 1205\\
     0.2 & 0.01 & $1753\pm 8$ & 2773\\
     0.1 & 0.01 & $5420\pm 10$& 8415\\
    \end{tabular}
  \end{center}
  \caption{The expected number of calls to \tpa{} for given 
           $(\epsilon,\delta)$.  Based off of 1000 simulations.  Times
           reported as mean of sample plus or minus standard deviation
           of sample.}
  \label{TBL:runtimes}
  \end{table}
\end{center}

  Table~\ref{TBL:runtimes} shows the expected running time for the new
  algorithm versus the old, which used Chernoff inequalities to bound
  the tails of the Poisson distribution.  The improvements are in
  the second order, which is why as $\delta$ shrinks relative to $\epsilon$,
  the improvement is lessened.  Still, for reasonable values of 
  $(\epsilon,\delta)$, the improvement is very noticeable.

\begin{example}
  Consider the Ising model~\cite{ising1925}, 
  where each node of a graph with vertex set $V$
  and edge set $E$ is assigned either a 0 or 1.
  For a configuration $x \in \{0,1\}^V$, let 
  $H(x) = \#\{e=\{i,j\} \in E:x(i) = x(j)\}$.  Then say that $X$ is a
  draw from the Ising model if $\prob(X = x) = \exp(\beta H(X))/Z(\beta)$,
  where $Z(\beta) = \sum_{y \in \{0,1\}^V} \exp(\beta H(y))$ is known as
  the {\em partition function}

  The goal is to find the partition function for various values of $\beta$.
  Note that $Z(0) = 2^{\#V}$ is known, so finding $Z(\beta)/Z(0)$ is sufficient
  to find $Z(\beta)$.

  Considering the Ising model on the $4 \times 4$ square lattice with 
  16 nodes in order to keep the numbers reasonable.  Then 
  $Z(1) \approx 3.219 \cdot 10^{11}$ and 
  $\ln(Z(1)/Z(0)) \approx 15.40.$  The method for using \tpa{} on a Gibbs
  distribution is found on p. 99 of~\cite{huber2014a}.  Methods for 
  generating samples from the Ising model for use in \tpa{} abound.  See for 
  instance~\cite{proppw1996,miramr2001,swendsenw1986,huber2003a}.  As long
  as $\beta$ is not too high, these methods are very fast.

  Using 100 calls with $(\epsilon,\delta) = (0.2,0.01)$ gives an estimate
  of $5200\pm70$ for the number of calls needed with the new Poisson estimate,
  while the old method requires 23249, making the new approach 
  over 4 times as fast in this instance for the same error guarantee.
\end{example}


\end{document}